\documentclass[%
 reprint,
superscriptaddress,
 amsmath,amssymb,
 aps,
]{revtex4-2}
\usepackage{amsmath,amsthm,amssymb,amstext,physics,float,color} % Lots of math symbols and environments
\usepackage{marvosym}

\usepackage{graphicx}% Include figure files

\newtheorem*{theorem-non}{Adversarial Strategy}
\usepackage{dcolumn}% Align table columns on decimal point
\usepackage{bm}% bold math
\begin{document}

\preprint{APS/123-QED}

\title{Controlled Quantum  Teleportation in the Presence of an Adversary}% Force line breaks with \\

\author{Sayan Gangopadhyay}
\email{sgangopa@uwaterloo.ca}
\affiliation{Institute for Quantum Computing, University of Waterloo, Waterloo, Canada}
\affiliation{ Department of Physics and Astronomy, University of Waterloo, Waterloo, Canada}
\author{Tiejun Wang}
\affiliation{School of Science, Beijing University of Posts and Telecommunications, Beijing 100876, China}
\author{Atefeh Mashatan}
\affiliation{Ted Rogers School of Information Technology Management, Toronto Metropolitan University, Toronto, Canada}
\author{Shohini Ghose}%
 \email{sghose@wlu.ca}
\affiliation{Institute for Quantum Computing, University of Waterloo, Waterloo, Canada}
\affiliation{ Department of Physics and Astronomy, University of Waterloo, Waterloo, Canada}
\affiliation{Department of Physics and Computer Science, Wilfrid Laurier University, Waterloo, Canada}

\date{\today}% It is always \today, today,

\begin{abstract}
%Controlled Quantum Teleportation (CQT) is a paradigmatic multipartite quantum communication protocol where a third party determines the success or failure of quantum teleportation. In situations where the exact characterization of a quantum device is not known, the theoretical framework of Device Independent (DI) quantum information processing is useful. 
We present a device independent analysis of  controlled quantum teleportation where the receiver is not trusted. We show that the notion of genuine tripartite nonlocality allows us to certify control power in such a scenario. By considering a specific adversarial attack strategy on a device characterized by depolarizing noise, we find that control power is a monotonically increasing function of genuine tripartite nonlocality. These results are relevant for building practical quantum communication networks and also shed light on the role of nonlocality in multipartite quantum information processing.

\end{abstract}

\maketitle

%\tableofcontents

\section{\label{intro}Introduction}
%Quantum teleportation is arguably one of the most important protocols in quantum information science. 
Controlled Quantum Teleportation (CQT) is a paradigmatic multipartite quantum communication protocol where a third party (or parties) determines the success or failure of quantum teleportation~\cite{karlsson1998quantum}.
In the standard quantum teleportation protocol \cite{PhysRevLett.70.1895}, an arbitrary qubit is teleported from Alice to Bob using maximally entangled EPR states. This idea has found applications in several quantum information schemes such as quantum gate teleportation \cite{gottesman1999demonstrating}, cluster state quantum computing \cite{raussendorf2001one} and quantum repeaters \cite{sangouard2011quantum}. Controlled quantum teleportation is an important multipartite extension of bipartitie quantum teleportation  \cite{karlsson1998quantum}.  CQT is an integral part of the design of Quantum Teleportation Networks (QTN) that could form the backbone of the future quantum internet \cite{castelvecchi2018quantum}.\par In the simplest CQT scheme, the success or failure of teleportation of a qubit from the sender to the receiver can be decided by a third party, also known as the controller, Charlie \cite{karlsson1998quantum}. This scheme works under the assumption that all parties are trusted, i.e., their devices and measurements are working exactly according to the specifications provided by the supplier. However, this is often not the case in real applications. There may be physical imperfections in the devices and measurement directions, which can not only lead to a lower teleportation fidelity but also can be used by an adversary to control the success or failure of teleportation. In scenarios where the precise knowledge of the underlying quantum state and measurement settings is not available, the framework of device-independent (DI) quantum information processing  can be useful \cite{pironio2009device}. \par
Traditionally, the notion of Bell nonlocality has been used for DI quantum information applications \cite{bellnonlocality}. Bell nonlocality refers to the phenomenon in which correlations obtained by performing local measurements on distant entangled states cannot be explained by any local hidden variable (LHV) model \cite{bell1964}. Bell nonlocality is often quantified by the violation of Bell inequalities. In applications such as DI Quantum Key Distribution (DIQKD) \cite{diqkd}, the violation of a Bell inequality strictly bounds the information leaked to an eavesdropper. Other notable applications of DI analysis include randomness amplification \cite{pironio2010random} and obtaining computational advantages \cite{anders2009computational}. What makes the DI formulation so powerful is the minimalism in assumptions regarding the functioning of the devices involved; it relies only on the observed measurement outcome correlations.\par
Recently, a stronger form of Bell's inequality has been used in the DI study of a multipartite quantum communication protocol called secret sharing \cite{DIsecret}. A secret  bit is split among $n-1$ parties in such a way that at least $m\leq n-1$  parties must collaborate to reveal it. Moreover, upto $n-2$ parties can be untrusted. It was shown that the maximal violation of Svetlichny's inequality \cite{svetlichny} guarantees that no information is leaked to an untrusted receiver. For non-maximal violation, it was observed that the larger the Svetlichny inequality violation, the smaller is the information leaked to the untrusted parties. 
 In \cite{DIsecret}, the secret message was a classical bit whereas a quantum secret sharing protocol can also share qubits. In the tripartite case, a quantum secret sharing (QSS) protocol becomes equivalent to a CQT protocol where Alice's secret qubit can be recovered by Bob iff Charlie assists \cite{pirandola2015advances}. This naturally leads one to ask what is the role and nature of nonlocality in the device independent (DI) analysis of QSS where the secret is quantum. A first step towards that goal would be to study DI tripartite QSS or equivalently, DI CQT. While DI quantum teleportation of a qubit was published close to a decade ago \cite{ho2013device}, a DI study of CQT is still an open question. \par
 In this work, we consider a CQT scenario where the receiver is untrusted. The receiver can collude with an external eavesdropper to increase the fidelity of teleportation beyond the classical limit even when the controller has not allowed it. Despite the allowed collusion and depending on the nature of the device used, it might still be impossible for Bob  to  achieve  a  teleportation  fidelity  as  high  as  that with Charlie’s participation.  Therefore, an untrusted receiver may lead to a decrease in control power but not entirely nullify it. We quantify the effective control power in this scenario and show that the maximal violation of the Svetlichny inequality guarantees maximum control power. We also show that the effective control power is non-zero for a significant range of the Svetlichny inequality violation where it is monotonically increasing. Finally, we emphasize the necessity of genuine tripartite nonlocality in the DI certification of the CQT scheme. We discuss the relevance and importance of these results in Section \ref{summary}.

\section{Background}
\subsection{Control Power in CQT} \label{cqt}
The controlled quantum teleportation protocol is similar to the standard quantum teleportation protocol with an added step. The objective is to teleport an unknown qubit from Alice to Bob only with the permission of the controller, Charlie. 
\begin{enumerate}
\item Let Alice, Bob and Charlie share a GHZ state of the form $(\ket{000}_{ABC}+\ket{111}_{ABC})/\sqrt{2}$. It can be equivalently expressed as $(\ket{\phi^+}_{AB}\ket{+}_{C}+\ket{\phi^-}_{AB}\ket{-}_C)/\sqrt{2}$. Henceforth, the subscripts $A,B,C$ will be dropped. 
\item Charlie performs a projective measurement in the $\sigma_X$ basis and gets an outcome $\gamma\in\{\pm 1\}$. His measurement operators are given by:
\begin{equation}M_C^\gamma=\frac{(\mathbb{I}_2+\gamma \sigma_X)}{2}; \ \ \ \ \ \gamma\in\{\pm 1\}.\end{equation}
After Charlie's measurement, Alice and Bob's joint state can be written as
\begin{align}
\rho_{\gamma}^{AB}&=tr^C(\rho_{GHZ}.(\mathbb{I}\otimes\mathbb{I}\otimes M_C^\gamma)) \\
&=\ket{\phi^{\gamma}}\bra{\phi^{\gamma}}.
\end{align}
\item Let the arbitrary state to be teleported be
\begin{equation}
    \rho^a=\frac{\mathbb{I}_2+\Vec{a}.\Vec{\sigma}}{2}. 
\end{equation}
Alice performs a Bell state measurement on the first two qubits of the state $\rho^a\otimes\rho^{AB}_{\gamma}$. The Bell state measurement is described by the following measurement operators: 
\begin{equation}\label{Bell}
M^A_{c_0c_1}=\ket{\phi^{c_0c_1}}\bra{\phi^{c_0c_1}},
\end{equation}
where $$\ket{\phi^{00}}=\frac{\ket{00}+\ket{11}}{\sqrt{2}}; \ket{\phi^{01}}=\frac{\ket{00}-\ket{11}}{\sqrt{2}};$$ \\ $$ \ket{\phi^{10}}=\frac{\ket{01}+\ket{10}}{\sqrt{2}}; \ket{\phi^{11}}=\frac{\ket{01}-\ket{10}}{\sqrt{2}}.$$
After Alice's measurement, Bob's state is projected into 
\begin{equation}
    \rho^B_{c_0c_1}=\frac{\mathbb{I}_2+(R_{c_0c_1\gamma}\Vec{a}).\Vec{\sigma}}{2},
\end{equation}
where 

\begin{equation}\label{rot}
    \begin{aligned}
    R_{0,0,+1}&=\mathbb{I}_3;R_{0,1,+1}=R_z(\pi);\\R_{1,0,+1}&=R_x(\pi);R_{1,1,+1}=R_y(\pi);\\R_{0,0,-1}&=R_z(\pi);R_{0,1,-1}=\mathbb{I}_3;\\R_{1,0,-1}&=R_y(\pi);R_{1,1,-1}=R_x(\pi).
    \end{aligned}
\end{equation}
$R_x(\pi), R_y(\pi), R_z(\pi)  $ represent rotations by $\pi$ along the three orthogonal directions.
\item Now Bob performs an $R_{c_0c_1\gamma}^{-1}$ rotation to retrieve $\rho^a$, 
\begin{equation}\label{finalwp}
\rho^B=\frac{\mathbb{I}_2+(R_{c_0c_1\gamma}^{-1}.R_{c_0c_1\gamma}\Vec{a}).\Vec{\sigma}}{2} =\frac{\mathbb{I}_2+\Vec{a}.\Vec{\sigma}}{2}=\rho^a. \end{equation}
\end{enumerate}
The average fidelity of teleportation between the unknown pure state ($\rho^a$) and Bob's final state ($\rho^{B}$) is calculated as \begin{equation}\label{fid}
    F=\int \frac{d\Vec{a}}{4\pi} \bra{a}\rho^B\ket{a}.
\end{equation} 
where the averaging has been performed over all pure qubit states.
Henceforth, we will refer to the average fidelity of teleportation performed with Charlie's participation as $F_C$ and that without Charlie's participation as $F_{NC}$.
From Eq. (\ref{finalwp}) and Eq. (\ref{fid}),
\begin{equation}
    F_{C}=\int\frac{d\Vec{a}}{4\pi} \bra{a}\frac{\mathbb{I}_2+\Vec{a}.\Vec{\sigma}}{2}\ket{a}=1.
\end{equation}
Suppose Charlie does not reveal $\gamma$. Bob will randomly perform either $R_{c_0c_1,+1}^{-1}$ or $R_{c_0c_1,-1}^{-1}$. In that case, the average teleportation fidelity is given by
\begin{widetext}
 \begin{equation}\label{FNCbs}
    F_{NC}=\int \frac{d\Vec{a}}{4\pi}\sum_{\gamma,\gamma'\in\{1,-1\}}\sum_{c_0,c_1\in\{0,1\}}P(\gamma,\gamma',c_0,c_1)\bra{a} \frac{\mathbb{I}_2+(R_{c_0c_1\gamma}^{-1}.R_{c_0c_1\gamma'}\Vec{a}).\Vec{\sigma}}{2}\ket{a}=\frac{2}{3}.
\end{equation}
\end{widetext}
% where $P(\gamma,\gamma',c_0,c_1)=P(c_0,c_1|\gamma,\gamma')P(\gamma,\gamma')$;
% $P(\gamma,\gamma')=P(\gamma|\gamma')P(\gamma')$
% \begin{equation}
%     P(\gamma')=\text{tr}(\rho_{GHZ}.(\mathbb{I}\otimes\mathbb{I}\otimes M_C^{\gamma'}))=\frac{1}{2}
% \end{equation}
% \end{widetext}
% \begin{equation}
%     P(\gamma|\gamma')=P(\gamma')=\frac{1}{2}
% \end{equation}
% therefore,
% \begin{equation}
%     P(\gamma,\gamma')=\frac{1}{2}.\frac{1}{2}=\frac{1}{4}
% \end{equation}
% \begin{equation}
%     P(c_0,c_1|\gamma,\gamma')=\text{tr}(\rho_{\gamma}^{AB}M^A_{c_0c_1})=\frac{1}{4}
% \end{equation}
% Substituting the above conditional probabilities in Eq. (\ref{FNCbs}), we get
% \begin{widetext}
% \begin{equation}
%     F_{NC}=\int \frac{d\Vec{a}}{4\pi}\sum_{\gamma,\gamma'\in\{1,-1\}}\sum_{c_0,c_1\in\{0,1\}}\frac{1}{16}\bra{a} \frac{\mathbb{I}_2+(R_{c_0c_1\gamma}^{-1}.R_{c_0c_1\gamma}\Vec{a}).\Vec{\sigma}}{2}\ket{a}=\frac{2}{3}
% \end{equation}
% \end{widetext}
The control power of Charlie is defined to be the difference between $F_C$ and $F_{NC}$ \cite{pract}. Without the controller's participation, the average teleportation fidelity should be minimized and with the controller's participation it should be maximized. Thus, the higher the difference between the above fidelities, the higher is the control. The control power of Charlie is therefore expressed as follows:
\begin{equation}
    CP=F_C-F_{NC}=1-\frac{2}{3}=\frac{1}{3}.
\end{equation}
\subsection{Device Independent Certification of Quantum Resources Used in Teleportation}\label{diqt}
In the usual two-party teleportation protocol using EPR states, Bancal et al. \cite{Ho} suggested a construction by which it is possible to device-independently certify whether a teleportation device is using quantum resources. Alice and Bob are given a pair of black boxes that supposedly perform the quantum teleportation of a qubit. The input of each box is a unit vector of the Bloch sphere. Alice's box takes the state to be teleported as the input given by the vector $\Vec{a}$ and outputs two bits $(c_0c_1) \in\{0,1\}^2$. In the back box scenario of teleportation, it is assumed that Alice knows the state to be teleported. Bob's box takes a vector $\Vec{b}$ as input and gives one bit output $\beta \in \{+1,-1\}$. $\Vec{b}$ represents a measurement of the teleported state in the $\Vec{b}.\Vec{\sigma}$ direction. The vendor claims that their boxes contain EPR states, using which the teleportation is performed on each run. Alice and Bob wish to verify the vendor's claim.  %Ideally, Bob's black box must contain a qubit in the state $
    %\rho^B_{c_0c_1}=\frac{\mathbb{I}_2+(R_{c_0c_1}\Vec{a}).\Vec{\sigma}}{2}$ where $R_{00}=\mathbb{I}_2;R_{01}=R_X(\pi);R_{10}=R_Y(\pi);R_{11}=R_Z(\pi)$. 
    
    In \cite{Ho}, it was shown that it is indeed possible to certify teleportation in this black  box scenario. In other words, it is possible to infer a posteriori that the black boxes used quantum resources to perform the teleportation from the input/output statistics of the boxes.\par
    
    \par
    In a teleportation protocol, Alice is typically required to send $(c_0,c_1)$ to Bob such that he can perform the appropriate corrective rotation based on some pre-established agreement to prepare the unknown qubit. However, it was shown \cite{toner} that 2 classical bits of information are enough to simulate the statistical correlations of the maximally entangled singlet state. Hence, a black box certification of quantum resources is not possible if Alice communicates with Bob. Therefore, Alice and Bob do not reveal any of their measurement input/output until several (ideally infinite) rounds have been completed. Each round consists of the pair of boxes taking inputs $(\Vec{a},\Vec{b})$  and giving outputs $(c_0c_1,\beta)$. In the end, they can construct the probability distributions $P(c_0c_1,\beta|\Vec{a},\Vec{b})$. Alice and Bob have no knowledge of the inner working of the composite black boxes except that they are not communicating. It is also assumed that Alice and Bob have free will. At the end of the protocol, the only information they have is the  data table of inputs and outputs ($c_0c_1,\beta|\Vec{a},\Vec{b}$) for each round. Using this, they must test whether the source of correlations is quantum. This scenario can be easily mapped into a Bell scenario.  \par
    Let Alice choose from state settings $\{\Vec{a_1},\Vec{a_2}\}$ and Bob choose from measurement setting $\{\Vec{b_1},\Vec{b_2}\}$. Alice's black box gives two bits $(c_0,c_1)$ as the output, which must be mapped to one bit. This can be done by choosing $\alpha=2c_j-1$ if Alice's input is $a_j$. Thus, one can construct the distributions $P(\alpha,\beta|j,k)$ from $P(c_0c_1,\beta|\Vec{a},\Vec{b})$. Finally, the Clauser-Horne-Shimony-Holt (CHSH) Bell function can be calculated as \begin{equation}CHSH=\sum_{j,k\in\{0,1\}}P(\alpha=\beta|j,k)-P(\alpha\neq\beta|j,k).\end{equation}
    If $CHSH>2$, it is guaranteed that the pair of black boxes generate statistical correlations that are not possible to obtain using only classical resources \cite{bell1964}. In other words, if $CHSH>2$, it is certified that quantum resources are being used for teleportation.
    
\section{Device Independent Controlled Teleportation of a Qubit to an Untrusted Receiver}\label{dicqt}
    Here we consider the controlled teleportation scenario where a qubit in the possession of Alice (sender) is teleported to Bob (receiver) only when Charlie (controller) participates. However, we assume that Bob is not a trusted party. Bob can collude with external agents henceforth labelled Derek to extract extra information that can increase the fidelity of teleportation beyond the classical limit of $2/3$ even without Charlie's participation. The presence of an untrusted receiver does not necessarily imply that the controller has no impact on the teleportation fidelity. In spite of the allowed collusion and depending on the nature of the device used, it might still be impossible for Bob to achieve a teleportation fidelity as high as that with Charlie's participation. 
    %Therefore, an untrusted receiver can at most lead to a decrease in control power.
    In this section, we address the question of how to device independently certify control power in such a scenario.   \par 
    The goal of a CQT protocol is to ensure that Charlie can control whether quantum resources are being used for teleportation. Hence it is necessary to first certify that the given device is capable of using quantum resources for teleportation.  We will present a construction using which it is possible to certify whether quantum resources are being used for controlled teleportation to an untrusted receiver. 
   
The idea of certification of quantum resources used in black box teleportation described in Section \ref{diqt} can be adapted for black box \textit{controlled} teleportation. \par
 \subsection{Scenario}\label{scenario}
         Derek supplies three black boxes to Alice, Bob and Charlie that can allegedly perform controlled quantum teleportation of a qubit such that Charlie is the controller, Alice is the sender and Bob is the receiver. The ideal functions of the three black boxes are as follows:
        \begin{enumerate}
            \item Charlie's black box: Accepts a measurement setting $\Vec{c}\in S^2$ as input and upon measurement gives an outcome $\gamma\in\{+1,-1\}$. 
            \item Alice's black box: Accepts the state to be teleported $\Vec{a}\in S^2$ as input and upon Bell-measurement [Eq. (\ref{Bell})] gives an outcome $s_0s_1\in\{0,1\}^2$.
            \item Bob's black box: Performs a corrective rotation $R_{s_0s_1\vec{c}\gamma}^{-1}$ which is specified in the instruction manual of the device. Accepts a measurement setting $\Vec{b}\in S^2$ as input and upon measurement gives an outcome $\beta \in\{+1,-1\}$. 
        \end{enumerate}

    In an ideal CQT scheme, the shared state of Alice, Bob and Charlie's black boxes is the GHZ state $(\ket{000}+\ket{111})/\sqrt{2}$. The analysis of this ideal scheme has been shown in Section \ref{cqt}. The control power of Charlie for the ideal CQT scheme is $CP=\frac{1}{3}$.\par 
However, we do not assume anything apriori regarding the inner working of the black boxes - including the underlying composite state of Alice, Bob and Charlie and their individual measurement basis. Additionally, we do not trust the receiver Bob. Therefore, Alice and Charlie are \textit{trusted} and Bob is \textit{untrusted}.
 
\subsection{The Adversary}\label{adversary}
In the above scenario \ref{scenario}, let Bob be the \textit{untrusted} part and Derek be the \textit{eavesdropper}. Together, they play the role of an \textit{adversary}. \par
\subsubsection{Adversarial Goal}\par
The goal of the adversary is to maximize the average teleportation fidelity when Charlie has not allowed the teleportation, i.e., not revealed the input setting $\Vec{c}$ and the outcome $\gamma$ of his measurement. \par
\subsubsection{Adversarial Capabilities}\label{capabilities}
\begin{enumerate}
    \item The eavesdropper Derek is restricted to acting on individual signals separately. 
    \item Derek can generate the shared states of the black boxes, i.e., hold the common source of correlations but has no direct access to the input/output variables of the three parties.
    \item The untrusted party Bob, can correlate his inputs with that of the common source of correlations held by Derek. 
    \item Derek and Bob can jointly extract another outcome $\delta$ that can be potentially used to increase the fidelity of teleportation. This means that the corrective rotation performed by Bob can also depend on $\delta$. This extra capability of the adversary allows us to write Bob's corrective rotations as $R_{s_0s_1\vec{c}\gamma\delta}^{-1}$ in contrast to $R_{s_0s_1\vec{c}\gamma}^{-1}$ in the ideal scenario. 
\end{enumerate}
\subsection{Device Independent Test of Quantum Resources}
%Consider a restriction to the ideal scenario described in Section 
In order to account for the adversary, we consider here a modification of the ideal scenario of Section \ref{scenario}:
\subsubsection{Device Independent Test Scenario}\label{DIscenario}
\begin{enumerate}
 \item Derek is the manufacturer and supplier of the black boxes of Alice, Bob and Charlie.
 \item Alice, Bob and Charlie's black boxes cannot communicate with each other.

            \item Charlie's black box: Accepts a measurement setting out of two distinct choices $\Vec{c}\in \{\Vec{c_0},\Vec{c_1}\}$ as input and upon measurement gives an outcome $\gamma\in\{+1,-1\}$. 
            \item Alice's black box: Accepts the state to be teleported out of two distinct states $\Vec{a}\in \{\Vec{a_0},\Vec{a_1}\}$ as input and upon Bell-measurement [Eq. (\ref{Bell})] gives an outcome $s_0s_1\in\{0,1\}^2$.
            \item Bob's black box: Performs a corrective rotation $R_{\delta}^{-1}$ given to him by Derek. Accepts a measurement setting out of two distinct choices $\Vec{b'}\in \{\Vec{b'_0},\Vec{b'_1}\}$ as input and upon measurement gives an outcome $\beta \in\{+1,-1\}$. Equivalently, the corrective rotation can be included in the measurement of Bob such that the new measurement choices are $\Vec{b}=\{\Vec{b_0}=R_{\delta}^{-1}\Vec{b'_0}, \ \ \Vec{b_1}=R_{\delta}^{-1}\Vec{b'_1}\}$. 
            \item Alice and Charlie independently choose their individual measurement setting. Their choice of measurement setting depends only on their free-will. However, Bob's choice of measurement setting can be influenced by Derek. Therefore, the untrusted receiver does not have free-will.   
            \item  We require that the announcement of inputs and measurement outcomes be made simultaneously by all parties after several rounds of the experiment have been completed. This ensures that the inputs or outcomes of the trusted parties are not used to the advantage of the untrusted party.
            \item The announcements at the end of several rounds of sending in inputs and recording outcomes from each black box will be a table of $(s_0s_1,\beta,\gamma)$ given $(j,k,l)$ for each round where $j, k, l$ denote the input setting (0 or 1) of Alice, Bob and Charlie respectively. From this data, the joint probability distribution $p(s_0s_1,\beta,\gamma|j,k,l)$ can be computed.
        \end{enumerate}
For the purpose of device independent testing, we require that Alice's output  $s_0s_1$ be mapped into a single bit $\alpha$ in the following way:
$$\alpha=2s_j-1 \ , \ \  \ \ \text{where Alice's input is }\Vec{a_j}.$$ Using this map, one can construct the distributions $P(\alpha,\beta,\gamma|j,k,l)$ from $P(s_0s_1,\beta,\gamma|j,k,l)$.
%\begin{remark}
We note that the choice of input merely indicates that each party can choose to press one out of two buttons. It is assumed that the parties do not know which measurement basis the buttons correspond to. 
%\end{remark}
\subsubsection{Causal Structure of the Device Independent Test Scenario}
The formalism of Directed Acyclic Graphs (DAG)  \cite{Chaves2017causalhierarchyof} will be used in later sections to characterize the correlations generated in the DI scenario of CQT with untrusted receiver. The causal structure of inputs, outcomes and common source of correlations of the different parties involved can be represented in the following way using directed acyclic graphs (DAG):
\begin{figure}[H]
        \centering
        \includegraphics[scale=0.5]{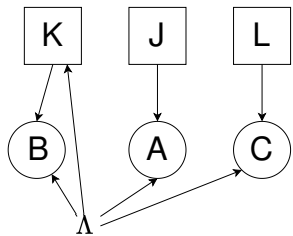}
        \caption{DAG representation of the causal model where $\Lambda$ is the common source of correlations supplied by Derek. In a classical model, $\Lambda$ is a shared random variable whereas in a quantum model, $\Lambda$ is a potentially entangled quantum state $\rho$. Here J, K, L denote the inputs of Alice, Bob and Charlie respectively. A, B, C denote the outcomes of Alice, Bob and Charlie respectively. Note that, since Bob is untrusted, he can correlate his inputs with that of the common source of correlations. }
        \label{DAG}
    \end{figure}
In Fig. \ref{DAG}, each directed edge represents a causal relation between two nodes. The start of each edge is called the parent node and the arrival of each edge is called the child node. In a classical causal model, the source of correlations and all other nodes are random variables such that the parent nodes completely characterize a child node. \par
The explanation of how to interpret the DAG in Fig. \ref{DAG}, is as follows:\\
\begin{itemize}
    \item The node $\Lambda$ represents the common source of correlations (the underlying state in the black boxes) supplied by the eavesdropper Derek. 
    \item J, K, L denote the inputs (measurement choices) of Alice, Bob and Charlie respectively. A, B, C denote the outcomes $\alpha,\beta,\gamma$ of Alice, Bob and Charlie respectively. 
    \item The arrows from $\Lambda$ to A, B and C represent the fact that the outputs of Alice, Bob and Charlie depend on the underlying state of their black boxes. 
    \item The arrows from J, K, L to A, B, C respectively represent the dependence of the outcomes of each party on the choice of their measurement setting.
    \item There are no arrows from the input or outcome of one party to another because their black boxes are not allowed to communicate with each other (Section \ref{DIscenario}).
    \item Since the inputs of the trusted parties Alice and Charlie can be independently chosen (measurement setting choice independence/free-will), J and L have no parent nodes. 
    \item Since Bob is untrusted, his input may be correlated with that of the common source of correlations. Hence, Bob does not have measurement independence. The arrow from $\Lambda$ to K represents the fact that Bob's choice of measurement setting depends on the information he receives from Derek through the black box. 
    \item Extra arrows from inputs/outcomes from one party to another were avoided due to the requirement that all announcements be made simultaneously at the end of several rounds.
\end{itemize}

If $\{v_i\}_{i=1}^n$ denote the nodes of the graph, $p(v_k|v_1,v_2,...,v_n)=p(v_k|pa(v_k))$ where $pa(v_k)$ denotes all parent nodes of $v_k$. Thus, the joint probability distribution of all the nodes is given by $p(v_1,v_2,...,v_n)=\Pi_{v_i}p(v_i|pa(v_i))$. \par

Therefore, in the causal structure with an untrusted receiver (Fig. \ref{DAG}), if $\Lambda$ is a classical random variable then the joint probabilities must admit the following decomposition:
    \begin{equation}\label{localdecomp}
    p(\alpha\beta\gamma|jkl)=\sum_\lambda p(\alpha|j\lambda)p(\beta|k\lambda)p(\gamma|l\lambda)\frac{p(k|\lambda)p(\lambda)}{p(k)}.
    \end{equation}

\subsubsection{Bell Inequality Characterizing the Device Independent Test Scenario}

In \cite{Chaves2017causalhierarchyof}, it was shown that the following DAGs are equivalent:\\
\\

\begin{figure}[H]
    \centering
    \includegraphics[scale=0.3]{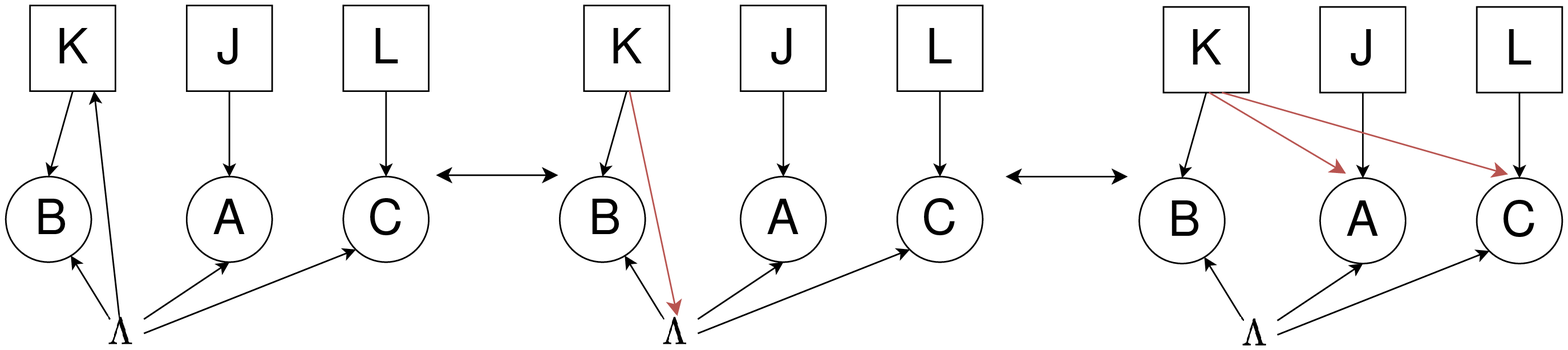}
    \caption{DAG Equivalence.}
    \label{dageq}
\end{figure}

This means that the  DAG representing the device independent scenario of controlled teleportation with an untrusted receiver (left-most DAG) is equivalent to the broadcasting scenario (right-most DAG). Broadcasting represents a scenario where one or more parties can openly communicate their choice of input to other parties. Here, in the right-most DAG it is seen that Bob communicates his choice of input to Alice and Charlie. \par
In \cite{bancal_2009}, the Bell inequalities characterizing the broadcasting scenario were studied. In the case of $n$ parties out of which  $n-m$ parties broadcast their inputs to all other parties and the remaining $m$ parties do not communicate their input to any other party, the tight bound to a specific Bell inequality called the Svetlichny inequality was found.

The Svetlichny function involves distributions of the form $p(\alpha\beta\gamma|jkl)$ which is obtained from the device independent test rounds.

\begin{equation}\label{S}
\begin{aligned}
S=& p_{A}(+1 \mid 0) \mathrm{CHSH}_{(+1)0}-p_{A}(-1 \mid 0) \mathrm{CHSH}_{(-1)0} \\
&+p_{A}(+1 \mid 1) \mathrm{CHSH}_{(+1)1}^{\prime}-p_{A}(-1 \mid 1) \mathrm{CHSH}_{(-1)1}^{\prime}.
\end{aligned}
\end{equation}
where $p_{A}(\alpha \mid j)$ is the marginal probability of Alice, $\mathrm{CHSH}_{\alpha j}$ and $\mathrm{CHSH}_{\alpha j}^{\prime}$ refer to symmetries of the Clauser-Horne-Shimony-Holt (CHSH) inequality given by $\mathrm{CHSH}_{\alpha j}=E_{(+1)0}^{\alpha j}+E_{(+1)1}^{\alpha j}+E_{(-1)0}^{\alpha j}-E_{(-1)1}^{\alpha j}$ and $\mathrm{CHSH}_{\alpha j}^{\prime}=E_{(+1)0}^{\alpha j}-E_{(+1)1}^{\alpha j}-E_{(-1)0}^{\alpha j}-$
$E_{(-1)1}^{\alpha j}$, and $E_{k l}^{\alpha j}=\sum_{\beta, \gamma=+1,-1}\beta\gamma  \ p(\beta, \gamma \mid j, \alpha, k, l)$ is the expectation value of the measurement outcome of Bob and Charlie conditioned on a given outcome $\alpha$ and the input $j$ of Alice.\par
The result from ref. \cite{bancal_2009} that we will use in this chapter is given by the following statement:

   \textbf{ For the broadcasting scenario where \  $n-m$ is odd,  \ \  \ 
    $|S|\leq2^{(n-m)/2+3/2}$. \\ Moreover, this bound is tight, i.e., there exists a classical strategy ($\Lambda$) such that $|S|=2^{(n-m)+3/2}$}.

In the broadcasting scenario of our interest (right-most DAG in Fig. \ref{dageq}), $n=3$ and $n-m=1$. Therefore, in this case, 
\begin{equation}\label{broadcasting}
   | S|\leq 4.
\end{equation}

Using this result [Eq. (\ref{broadcasting})] and the equivalence of DAGs (Fig. \ref{dageq}), we can immediately state the following:

\newtheorem{Theorem}{Theorem}
\begin{Theorem}\label{claim}
Any probability distribution $p(\alpha\beta\gamma|jkl)$ admitted by the DAG representing the device independent test scenario of controlled quantum teleportation with an untrusted receiver (as defined in Section \ref{DIscenario}) using classical strategies must satisfy the Svetlichny inequality $ | S|\leq 4$. Moreover, this bound is tight.
\end{Theorem}
\textbf{Condition for certification of quantum resources:} If the Svetlichny inequality  $ | S|\leq 4$ is violated, it is guaranteed that the observed probability distribution $p(\alpha\beta\gamma|jkl)$ was not generated entirely by classical means. Hence, the source of correlations ($\Lambda$) must be quantum.   

\subsubsection{Comparison with the Device Independent Scenario of Fully Trusted Controlled Quantum Teleportation}
In controlled quantum teleportation where all parties are trusted, it is easy to see that the device independent test scenario is represented by the following DAG:
\begin{figure}[H]
    \centering
    \includegraphics[scale=0.5]{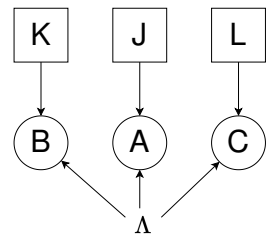}
    \caption{DAG for the device independent test scenario of fully trusted controlled quantum teleportation. Note that there are no arrows from $\Lambda$ to $K$ unlike that in the scenario of untrusted receiver considered earlier. Here, all parties have measurement choice independence and there is no communication between each other. }
    \label{dagtrusted}
\end{figure}
This device independent scenario is certified to be quantum if the obtained correlations violate Mermin's inequality \cite{mermin}. A maximal violation of Mermin's inequality implies that the underlying state is the GHZ state, therefore guaranteeing maximum control power. However, it was shown in Ref. \cite{DIsecret} that there exists a local model in the broadcasting scenario of Fig. \ref{dageq}, that violates Mermin's inequality maximally. Since the device independent scenario of controlled teleportation with an untrusted receiver is equivalent to the broadcasting scenario (Fig. \ref{dageq}), Mermin's inequality is not sufficient for its certification. \par
The presence of an untrusted part makes it necessary to use a stronger notion of nonlocality also known as genuine tripartite nonlocality/Svetlichny nonlocality for the certification of controlled teleportation.

%There is a well known result \cite{bancal_2009} which states that the distributions $p(\alpha\beta\gamma|ijk)$ admitted by the above DAG (Fig. \ref{dagtrusted})  must satisfy $|S|\leq 2$. However, we argued earlier in Claim \ref{claim} that for the untrusted receiver scenario, this bound is higher, i.e.,   $|S|\leq 4$. Since Svetlichny inequality is also a quantifier of nonlocality, we conclude that a stronger notion of nonlocality is required to perform the device independent test involving an untrusted receiver compared to that where all parties are trusted. This stronger notion of tripartite nonlocality is also known as genuine 3-way nonlocality \cite{genuine}.   
\subsubsection{When is the Controller's Authority Maximum?}

\begin{Theorem}\label{max authority}
Controller's authority is maximum with $CP=\frac{1}{3}$ in the controlled quantum teleportation scheme if $S=4\sqrt{2}$ is obtained from the device independent test.   
\end{Theorem}

\begin{proof}
Because of the normalization of $p_{A}(a \mid j)$, to achieve the maximum quantum violation $S=4 \sqrt{2}$ necessarily, we must have $\mathrm{CHSH}_{(-1)0}=\mathrm{CHSH}_{(-1)1}^{\prime}=-2 \sqrt{2}$
and $\mathrm{CHSH}_{(+1)0}=\mathrm{CHSH}_{(+1)1}^{\prime}=2 \sqrt{2}$.  
For any given $\alpha$ and $j$, Charlie and Bob's joint state must maximally violate the $\mathrm{CHSH}$. Only maximally entangled two-qubit pure states can generate such a correlation. \par
Thus, by the monogamy of entanglement, $\rho^{BCD}=\rho^{BC}\otimes\rho^{D}$ where $\rho^{BC}$ is a maximally entangled pure state of Bob and Charlie. This means that if $S=4\sqrt{2}$, then Derek cannot remain entangled to Bob and Charlie's devices.\par
Now suppose the same device which yielded $S=4\sqrt{2}$ during the device independent test, is used for the standard controlled quantum teleportation scheme (Section \ref{cqt}).   If Charlie does not give permission to teleport i.e., does not reveal $\gamma l$ then $p(\delta=\gamma l)=\frac{1}{4}$. This reflects the fact that Derek's state is separate from Bob and Charlie's composite state and therefore cannot obtain information about their systems by performing local operations on his system.  
Essentially, it means that Bob has no extra information about what the corrective rotations $R_{s_0s_1\gamma l}^{-1}$ are beyond what the trusted parties Alice and Charlie choose to reveal. \par
Alternatively, we note that only the maximally entangled GHZ state $\frac{\ket{000}+\ket{111}}{2}$ (upto local unitary operations) violates the Svetlichny inequality maximally, i.e., gives $S=4\sqrt{2}$. It directly follows that Derek's state must be separated from the tripartite state of Alice, Bob and Charlie since the maximally entangled GHZ state is a pure quantum state. The controller's authority for the GHZ state was derived in Section \ref{cqt} and  CP was shown to be  1/3. 
\end{proof}

\section{Controller's Authority from Non-maximal Violation of Svetlichny Inequality}\label{non-maximal}
In Section \ref{diqt}, we proposed a setup using which the Svetlichny function can be computed. We showed that a maximal violation of the Svetlichny Inequality [Eq. (\ref{broadcasting})] implies that Bob can get no useful information from Derek to increase the average teleportation fidelity without Charlie's permission. This ensures that the controller's authority is maximized. However, one can expect that the controller's authority will not be maximized in case of non-maximal violation of the Svetlichny Inequality. \par
In this section, we study the quantitative relationship between genuine 3-way nonlocality and the controller's authority using the example of two specific families of quantum states. We consider two noise models of the tripartite GHZ state - the total depolarizing channel and the qubit depolarizing channel - and compute the controller's authority as a function of Svetlichny Inequality [Eq. (\ref{broadcasting})] violation. We consider a specific attack strategy using which Bob can try to increase the average fidelity of teleportation without Charlie's permission. This strategy is based on the assumption that Alice, Bob, Charlie and Derek are restricted to performing quantum operations on their respective systems - no joint quantum operations allowed between any two parties. Charlie's authority will be quantified by the Effective Control Power (ECP) which we define as
\begin{equation}
    ECP=F_C^{NE}-F_{NC}^{E}.
\end{equation} Here $F_C^{NE}$ indicates the average fidelity of teleportation with Charlie's permission and no eavesdropping. $F_{NC}^E$ indicates the average fidelity of teleportation without Charlie's permission but with eavesdropping (i.e., Derek's participation). The general definition of controller's authority would be
\begin{equation}\label{CP}
CP=F_C^{E}-F_{NC}^{E}.
\end{equation}The superscripts ($E$) indicate that eavesdropping has been considered in both cases of control (with and without Charlie's permission). We shall leave the more general quantification of controller's authority as future work and focus on Effective Control Power $(ECP)$ of Charlie in this section. Note that $F_{C}^{E}\geq F_{C}^{NE}$ because eavesdropping will always increase fidelity. Therefore, $ECP$ is a lower bound to the more general definition of controller's authority quantified by [Eq. (\ref{CP})]. It is guaranteed that for the adversarial strategy under consideration, Charlie has at least $ECP$ amount of control. \par
%In the device independent test scenario of controlled teleportation experiment (Section \ref{dicqt}), we did not need to know how the input of each box was transformed into an output. However, when we have been given \ref{} for a specific quantum state, we must also specify what the measurement setting of each box is in order to compute a quantum probability distribution. We will assume that the measurement settings of Alice and Charlie and the corrective rotations of Bob are fixed according to the ideal CQT experiment with the GHZ state $\ket{\psi}=(\ket{000}+\ket{111})/\sqrt{2}$ as mentioned in \ref{dicqt}. 
\subsection{Adversarial Strategy}\label{derek strategy}
The ideal CQT scheme (Section \ref{cqt}) using the tripartite  GHZ state allows teleportation with perfect average fidelity with Charlie's permission. Alice and Bob's composite state is projected into one of the EPR states after Charlie performs a measurement and reveals the measurement setting and outcome to Bob. When Charlie does not make his measurement information known to Bob, i.e., does not give permission to teleport, Bob does not know which EPR state the teleportation channel has been projected into. This prevents Bob from achieving an average teleportation fidelity higher than $2/3$. Recall, that the goal of the adversary (Section \ref{adversary}) is to maximize the average teleportation fidelity when Charlie has not allowed the teleportation. Bob's  
next best option is to ask Derek to perform a measurement in a suitable basis and reveal its outcome ($\delta$) to Bob such that Alice and Bob's composite state is projected close to one of the four EPR states. This gives Bob some relevant information regarding the corrective rotations, using which he can recover the state to be teleported with higher fidelity (see point 4 of Adversarial Capabilities, Section \ref{capabilities}). Our goal is to find the maximum fidelity of teleportation that Bob can achieve without permission from Charlie using this extra information ($\delta$). This will allow us to evaluate the worst case Effective Control Power (ECP) of Charlie.\par

The adversarial strategy can thus be formulated as the following problem:

\begin{theorem-non}\label{strategy}
 Find the optimal POVM operators of Derek such that the fidelity between the bipartite post-measurement state of Alice and Bob and one of the EPR states is maximized. 
\end{theorem-non}
\subsubsection{Execution of Adversarial Strategy}
The adversarial strategy in Proposition \ref{strategy} can be executed in the following way:\\

Let $\{M_i\}$ denote Derek's POVMs, and $Pr(i)$ denote the probability of getting an outcome $i$. $\rho^{ABD}$ is the joint state of Alice, Bob and Derek. Note that Charlie's state is irrelevant in this case since he is not participating. \par 
Then the post measurement state of Alice and Bob is given by \begin{equation}
    \rho^{AB}_i=\frac{\Tr^D\left(\rho^{ABD}(\mathbb{I}\otimes\mathbb{I}\otimes M_i)\right)}{Pr(i)}.
    \end{equation}
Let $F(\rho_1,\rho_2)$ denote the fidelity between two density matrices $\rho_1$ and $\rho_2$. Then the problem of finding the optimal POVMs can be formulated as follows:
\begin{equation}\label{problem}
\begin{aligned}
&\text{Maximize} \ \ \ \ \sum_{i}Pr(i) F(\rho^{AB}_i,\ket{\phi_i}\bra{\phi_i})\\
&\text{Subject To}    \ \ \ \ M_i\geq 0, \ \sum_{i}M_i=\mathbb{I}
\end{aligned}
\end{equation}
\begin{equation}
\begin{aligned}
    &Pr(i)F(\rho^{AB}_i,\ket{\phi_i}\bra{\phi_i})\\&=\Tr(\Tr^D\left(\rho^{ABD}(\mathbb{I}\otimes\mathbb{I}\otimes M_i)\right)\ket{\phi_i}\bra{\phi_i})\\
   &=\Tr((\ket{\phi_i}\bra{\phi_i}\otimes\mathbb{I})\rho^{ABD}(\mathbb{I}\otimes\mathbb{I}\otimes M_i))\\
       &=\Tr(\Tr^D((\ket{\phi_i}\bra{\phi_i}\otimes\mathbb{I})\rho^{ABD})M_i).
    \end{aligned}
\end{equation}
Let $\tilde{\rho}_i$= $\Tr^D((\ket{\phi_i}\bra{\phi_i}\otimes\mathbb{I})\rho^{ABD})$. After substituting $\tilde{\rho}_i$ in Eq. (\ref{problem}), we get:
\begin{equation}\label{sdp}
\begin{aligned}
&\text{Maximize} \ \ \ \ \sum_{i}\Tr(\tilde{\rho}_i M_i)\\
&\text{Subject To}    \ \ \ \ M_i\geq 0, \ \sum_{i}M_i=\mathbb{I}, M_i=M_i^{\dagger}\\
&\text{Variable} \ \ \ \ \ M_i
\end{aligned}
\end{equation}

The optimization problem in Eq. (\ref{sdp}) can be cast into a semidefinite program (SDP) using the procedure given in Appendix \ref{AppendixA}. The SDP can be solved numerically using the \texttt{CVX} module \cite{cvx}, \cite{gb08} on MATLAB to obtain the optimal POVM measurements.

\subsection{Examples of Controller's Authority with Non-Maximal \\Svetlichny Inequality Violation}\label{examples}

A useful approach to analyze the adversarial strategy and Derek's role is to describe the problem using the framework of decoherence and noise models in open systems.  
The ideal CQT scheme assumes that the tripartite state shared by Alice, Bob and Charlie is a pure GHZ state ($\ket{\psi}=(\ket{000}+\ket{111})/\sqrt{2}$). More generally, in a realistic scenario, the pure GHZ state would undergo decoherence due to entanglement with the environment. This environment could be the adversary Derek, and hence the decoherence model can be used to analyze Derek's effect on CQT as follows.
%Often errors are introduced while distributing the qubits which leads to decoherence. 
The state $\rho_f$ after a decoherence process $\varepsilon$ is given by $\rho_f=\varepsilon\rho_i$. The action of $\varepsilon$ can be described using Kraus operators $\{E_{j}\}_j$ in the following way \cite{kraus1983states},\cite{zhang2013speed}:
 \begin{equation}
 \varepsilon \rho_i=\sum_{j=1}^{M} E_{j} \rho_i E_{j}^{\dagger}.
 \end{equation}

 First, we will consider the case where the whole state is affected by the same decoherence process described by the total depolarizing channel. It is a process in which the ideal GHZ state is mixed with white noise with probability $p$. Thus, \begin{equation}
     \rho_f^{total}=p\frac{\mathbb{I}}{8}+(1-p)\ket{\psi}\bra{\psi}.
 \end{equation} 
Next, we will consider the decoherence process in which each qubit of the tripartite GHZ state undergoes a depolarizing channel. Note that the qubit depolarizing process is physically more appropriate than the total depolarizing process because the three qubits are distributed to three different parties who are usually at different locations. Each qubit is coupled to its local environment and undergoes depolarizing independently though they were initially entangled at a single location. The qubit depolarizing channel is described by the Kraus operators $E_{0}=\sqrt{1-p^{\prime}} I ; \quad E_{i}=\sqrt{\frac{p^{\prime}}{3}} \sigma_{i}$ where $p^{\prime}=\frac{3p}{2}$ \cite{zhang2013speed}.   The final state after the qubit depolarizing process assuming the same depolarizing parameter for each quibit, is given by
 \begin{equation}
 \rho^{qubit}_f=\sum_{i j k} E_{i} \otimes E_{j} \otimes E_{k}\ket{\psi}\bra{\psi} \left[E_{i} \otimes E_{j}\otimes E_{k}\right]^{\dagger}.
 \end{equation}

 $\rho_f^{total},\rho^{qubit}_f$ are mixed states unless $p=0$. One can think of these mixed states as a part of a bigger pure quantum state $\ket{\psi^{ABCD}}$ comprising Alice, Bob, Charlie and Derek such that $\Tr^D(\ket{\psi^{ABCD}}\bra{\psi^{ABCD}})=\rho^{qubit/total}_f$. In the language of quantum cryptography, one would say Derek `holds' a purification of the mixed tripartite state. Since $\ket{\psi^{ABCD}}$ is a pure state, Derek is the most general eavesdropper. Any larger quantum system $\rho^{ABCD\Delta}$ with more eavesdroppers $\Delta_1,\Delta_2,\ldots,\Delta_N$ will be of the separable form $\rho^{ABCD\Delta}=\ket{\psi^{ABCD}}\bra{\psi^{ABCD}}\otimes \rho^{\Delta_1,\Delta_2,\ldots,\Delta_N}$, which means that $\Delta_1,\Delta_2,\ldots,\Delta_N$ cannot be correlated with $\ket{\psi^{ABCD}}$.
 \par
 In this paper, we will use the spectral decomposition  of a density matrix to obtain a purification of $\ket{\psi^{ABCD}}$ \cite{nie2011quantum} from $\rho_f^{total/qubit}$. 
 The purification $\ket{\psi_f^{total/qubit}}$ of $\rho_f^{total/qubit}$ is therefore given by
 \begin{equation}\label{purification}
     \ket{\psi_f^{total/qubit}}=\sum_{k\in \{1,2,\ldots,8\}}\frac{(\rho_f^{total/qubit}\otimes \mathbb{I}_8) \ket{\Lambda_k}\ket{\Lambda_k}}{\sqrt{\Tr(\ket{\Lambda_k}\bra{\Lambda_k}\rho_f^{total/qubit})}}.
 \end{equation}
 \subsubsection{Computing the Effective Control Power (ECP) for the Depolarized GHZ State}\label{ECP calculation}
We have all the necessary ingredients to compute $F_C^{NE }$ and $F_{NC}^{E}$ (defined in Section \ref{non-maximal}) for the total depolarized and qubit depolarized GHZ states.\par
Let $\rho^B_{s_0s_1\gamma l}$ denote Bob's state when Alice and Charlie's boxes have revealed the outputs  $s_0s_1$ and $\gamma l$ after performing the measurements $\{M_{s_0s_1}^{aA}\}_{s_0s_1}$ and $\{M_{\gamma l}^C\}_{\gamma l}$ respectively. 
$$\rho_{s_0s_1\gamma l}^B=\frac{\Tr^{aAC}\left(\left(\frac{\mathbb{I}+\Vec{a}.\Vec{\sigma}}{2}\otimes \rho_f^{total/qubit}\right)( M_{s_0s_1}^{aA} \otimes\mathbb{I}\otimes M_{\gamma l}^C )\right)}{\Tr\left(\left(\frac{\mathbb{I}+\Vec{a}.\Vec{\sigma}}{2}\otimes \rho_f^{total/qubit}\right)( M_{s_0s_1}^{aA} \otimes\mathbb{I}\otimes M_{\gamma l}^C )\right)}.$$
Then the average fidelity of teleportation with Charlie's permission is given by:
\begin{widetext}
 \begin{equation}
    F_C^{NE }=\int\frac{d\Vec{a}}{4\pi}\sum_{s_0,s_1,l \ \in \ \{0,1\}; \ \gamma \in \ \{\pm 1\}}P( s_0s_1 \gamma l)\bra{a} R^{-1}_{s_0s_1\gamma l}\rho^B_{s_0s_1\gamma l}(R^{-1}_{s_0s_1\gamma l})^{\dagger}\ket{a}.
\end{equation}
\end{widetext}

%\begin{remark}{\text{Fixed rotation matrices $R^{-1}_{s_0s_1\gamma l}$:}}
The calculation of $F_C^{NE }$ has been done assuming that the set of corrective rotations $\{R^{-1}_{s_0s_1\gamma l}\}_{s_0s_1\gamma l}$ of Bob are those that are required to exactly recover the qubit to be teleported, had the shared tripartite state ($\rho_f^{total/qubit}$) been the perfect GHZ state. In this paper, the set of rotation matrices have been assumed to be fixed. We have not considered variable rotations conditioned on the underlying tripartite state.
%\end{remark}
\begin{figure*}
    \centering
    \includegraphics[scale=0.58]{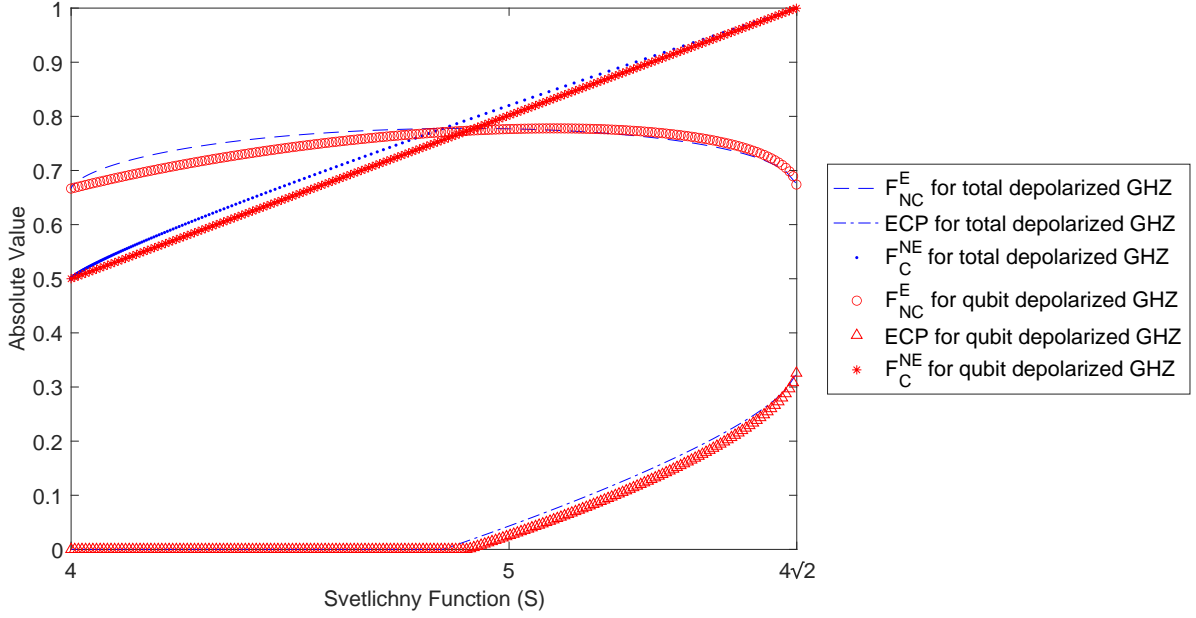}
     \caption{ Effective Control Power (ECP), average fidelity of teleportation with controller's permission ($F^{NE}_C$) and average fidelity of teleportation without controller's permission but with eavesdropper's participation ($F^E_{NC}$) as a function of maximum Svetlichny inequality violation (S) given by Eq. (\ref{max Svetlichny})  for both the qubit depolarized and total depolarized GHZ states with parameter $p \in (0,1)$ }.
\end{figure*}
To compute $F_{NC}^{E}$, we first need to determine the optimal POVMs of Derek which will execute the adversarial strategy (Section \ref{derek strategy}). \par
$\rho^{ABD}$ can be obtained from $\rho^{total/qubit}_f$ by first purifying it into $\ket{\psi_f^{total/qubit}}$ [Eq. (\ref{purification})] and then tracing out Charlie's system.
\begin{equation}
    \rho^{ABD}=\Tr^C\left(\ket{\psi_f^{total/qubit}}\bra{\psi_f^{total/qubit}}\right).
\end{equation}
By substituting $\rho^{ABD}$ in the expression  $\tilde{\rho}_i=\Tr^D((\ket{\phi_i}\bra{\phi_i}\otimes\mathbb{I}) \rho^{ABD})$, 
one can set up the optimization given in Eq. (\ref{sdp}).\par
Derek can then use the optimal POVMs $\{M_i\}_i$ to measure his system and reveal $\delta=i \in\{(+1)0,(+1)1,(-1)0,(-1)1\}$ to Bob. Hence, Alice and Bob's joint state ($\rho_i^{AB}$) conditioned on the result of Derek's measurement is given by:
\begin{equation}
    \rho_i^{AB}=\frac{\Tr^D\left(\rho^{ABD}( \mathbb{I}_2\otimes\mathbb{I}_2\otimes M_i)\right)}{\Tr\left(\rho^{ABD}( \mathbb{I}_2\otimes\mathbb{I}_2\otimes M_i)\right)}.
\end{equation}

Now Alice performs the Bell measurement $\{M_{s_0s_1}^{aA}\}_{s_0s_1}$ on the state to be teleported and her share of the composite state  $\rho_i^{AB}$. After Alice reveals her measurement outcome, Bob's qubit is prepared in the following state: 
\begin{equation}
    \rho^B_{s_0s_1i}=\frac{\Tr^{aA}\left(\left(\frac{\mathbb{I}+\Vec{a}.\Vec{\sigma}}{2}\otimes \rho_i^{AB}\right) \left(M_{s_0s_1}^{aA} \otimes\mathbb{I}\right)\right)}
    {\Tr\left(\left(\frac{\mathbb{I}+\Vec{a}.\Vec{\sigma}}{2}\otimes \rho_i^{AB}\right) \left(M_{s_0s_1}^{aA} \otimes\mathbb{I}\right)\right)}.
\end{equation}
Finally, the average fidelity of teleportation without Charlie's permission but with Derek's assistance ($F^{E}_{NC}$) can be computed. 
\begin{widetext}
\begin{equation}
    F^{E}_{NC}=\int\frac{d\Vec{a}}{4\pi}\sum_{s_0,s_1\in\{0,1\}; \ i\in \{0,1\}^2}P( s_0s_1 i)\bra{a} R^{-1}_{s_0s_1i}\rho^B_{s_0s_1i}(R^{-1}_{s_0s_1i})^{\dagger}\ket{a}.
\end{equation}
\end{widetext}

\subsubsection{Svetlichny Inequality Violation for Depolarized GHZ States}
We use the physical interpretation of the total depolarized GHZ state (TDGHZ) and the qubit depolarized GHZ state (QDGHZ) to derive the maximum Svetlichny inequality violation that can be obtained for a given depolarizing channel parameter. \par
As mentioned earlier, the total depolarized GHZ state can be seen as a probabilistic mixture of the perfect GHZ state $(\frac{\ket{000}+\ket{111}}{\sqrt{2}})$ and the completely mixed three qubit state ($\frac{\mathbb{I}}{8}$).  \par

The qubit depolarizing channel can be described by the process in which each qubit can get replaced by the completely mixed single qubit state ($\frac{\mathbb{I}}{2}$), with probability $p$. For the GHZ state, it means that with probability $(1-p)^3$ the state is unaltered; with probability $3p(1-p)^2$, the GHZ state is transformed into a bipartite entangled state  and with probability  $3p^2(1-p)+p^3$, it is transformed into a separable state. \par

It was shown in Ref. \cite{svetlichny} that tripartite entangled states are required for violating the Svetlichny Inequality $(S\leq 4)$. Bipartite entangled states and separable states do not violate this inequality while the tripartite entangled perfect GHZ state ($\frac{\ket{000}+\ket{111}}{\sqrt{2}}$) violates it maximally ($S=4\sqrt{2}$). Moreover, according to Theorem \ref{claim} it is possible to achieve the $(S = 4)$ bound using classical strategies. Therefore, the maximum Svetlichny inequality violation is given by the following equations:

\begin{align}\label{max Svetlichny}
    S_{TDGHZ}&=(1-p)S_{GHZ}+pS_{Classical}\\&=(1-p)4\sqrt{2}+4p.\\
    S_{QDGHZ}&=(1-p)^3S_{GHZ}+(1-(1-p)^3)S_{Classical}\\&=(1-p)^3 4\sqrt{2}+4(1-(1-p)^3).
    \end{align}
\subsection{Numerical Calculation of the Relationship Between Effective Control Power and Svetlichny Violation}

% \begin{figure}[H]
% \begin{subfigure}{1.0\textwidth}
% \centering
%     \includegraphics[width=\textwidth]{final_fig.png}
% \end{subfigure}
% \begin{subfigure}{0.4\textwidth}
%  \includegraphics[width=\textwidth]{qdep_trimmed.png}
%  \end{subfigure}
% \end{figure}

For a given depolarizing parameter $p\in (0,1)$, the maximum Svetlichny inequality violation was computed using Eq. (\ref{max Svetlichny}). The Effective Control Power (ECP) was computed using the method described in Section \ref{ECP calculation} and plotted against the maximum Svetlichny violation corresponding to the given value of parameter $p$. 
It is interesting to note that ECP is positive only when $S>4.84$ for the total depolarized GHZ state and when $S>4.90$ for the qubit depolarized GHZ state. The plots clearly show that ECP is a monotonically increasing function of $S$. The highest value of ECP is reached at the maximal violation of Svetlichny's inequality which confirms our Claim \ref{max authority}. It is important to note that the violation of the Svetlichny inequality [Eq. (\ref{broadcasting})] does not necessarily imply that Charlie has positive control power. For non-maximal violation, there is a small window in the range $4.84<S<4\sqrt{2}$ for the qubit depolarized GHZ states and $4.90<S<4\sqrt{2}$ for the total depolarized GHZ states, where ECP is positive, and hence Charlie maintains some level of control.

\section{Summary and Outlook}\label{summary}
In this paper, we have performed a device independent study of controlled teleportation of a qubit with an untrusted receiver. We constructed a device independently testable scenario in a way that allowed us to certify in the context of controlled teleportation, whether quantum resources were being used by the device despite the receiver being untrusted. We found in this case that the well-known Svetlichny inequality must be violated to certify quantum correlations. A maximal violation of the Svetlichny inequality guarantees maximum control power. This is in contrast to the controlled teleportation with all trusted parties where the maximal violation of Mermin's inequality was sufficient to certify maximum control power. This indicates that a stronger form of nonlocality, also known as `genuine tripartite nonlocality', is required to device independently test the controlled quantum teleportation with an untrusted receiver. Until recently \cite{DIsecret}, there was no application of higher order Bell's inequalities in DI quantum cryptography. Our work demonstrates one of the first instances of the usefulness of stronger forms of Bell's inequalities (Svetlichny inequality in this case) in DI quantum communication protocols.   \par
We proposed an adversarial strategy which, while not proven to be optimal, can effectively decrease the controller's authority by taking advantage of a non-ideal device that non-maximally violates Svetlichny inequality. By taking the example of two families of quantum states characterized by the total depolarized and the qubit depolarized GHZ states, we showed that the controller's authority is a monotonically increasing function of the maximal Svetlichny inequality violation. For the given family of depolarized GHZ states, adversarial strategy and a Svetlichny inequality violation, one can infer the controller's authority from our numerically obtained plot. We found a window of non-maximal Svetlichny inequality violation where the controller's authority is non-zero. This shows that the controlled teleportation scheme with an untrusted receiver is robust to depolarizing noise present in the device. This result can be a guiding framework for practical implementations of DI CQT.\par
In this paper we have used effective control power as the figure of merit in DI CQT, whereas the authors of  \cite{DIsecret} used $(1-p_{guess})$ (probability of the untrusted receiver guessing the wrong secret bit) as the figure of merit in the DI secret sharing of a bit.  However, we see that in both cases, the relationship of these figures of merit with that of genuine tripartite nonlocality is very similar. From this observation, intuition suggests that the von Neumann entropy of the quantum information revealed to the untrusted receiver would show similar trends to that of effective control power in DI CQT. The translation from effective control power to an entropic measure of quantum information would help to  generalize our work to a multipartite QSS protocol. \par 
Since CQT forms the basis of quantum teleportation networks, the techniques developed for DI analysis in this paper provide a stepping stone for performing DI analysis of quantum networks with untrusted nodes and practical implementation of a quantum internet.

%It would be imperative to investigate whether the eavesdropping strategy considered in Section \ref{examples} is optimal, i.e., whether there exists any better eavesdropping strategy for the depolarized GHZ state that would reduce the effective control power even further. Such an investigation would help to strengthen our findings.  }\par
%A practical way to extend this work would be to consider both the sender (Alice) and the receiver (Bob) to be untrusted parties. In that case, the criteria for device -independent certification will change. Svetlichny's Inequality violation may no longer be sufficient to confirm that the underlying source of correlations is quantum.
\begin{acknowledgments}
We thank N. L\"{u}tkenhaus, R. Mann for useful discussions. This research was funded by the Natural Sciences and Engineering Research Council
of Canada and NXM Labs Inc. NXM’s autonomous security technology enables devices, including
connected vehicles, to communicate securely with each other and their surroundings without human intervention
while leveraging data at the edge to provide business intelligence and insights. NXM ensures data
privacy and integrity by using a novel blockchain-based architecture which enables rapid and regulatorycompliant
data monetization. Toronto Metropolitan University is in the “Dish With One Spoon Territory.” The Dish
With One Spoon is a treaty between the Anishinaabe, Mississaugas and Haudenosaunee that bound them
to share the territory and protect the land. Subsequent Indigenous Nations and peoples, Europeans and all
newcomers, have been invited into this treaty in the spirit of peace, friendship and respect. Wilfrid Laurier
University and the University of Waterloo are located on the traditional territory of the Neutral, Anishnawbe and Haudenosaunee peoples. The University of Waterloo is situated on the Haldimand Tract, the land promised to the Six Nations that includes six miles on each side of the Grand River.
We thank them for allowing us to conduct research on their land.

\end{acknowledgments}

\appendix

\section{SDP Optimization}\label{AppendixA}
Consider the optimization program given in Eq. (\ref{sdp}). We can cast it into a semidefinite program in the following way:\\
Let $\mathcal{X}$ be a Hilbert space of dimension $D$, where $D$ is also the dimension of Derek's quantum system. Define $\Phi$ as a Hermitian preserving map $T(\mathcal{X},\mathcal{X})$. Then the optimization problem is given by:
\begin{equation}
\begin{aligned}
&\underline{\text { Primal problem }} \\
&\text{ maximize:}  \ \ \langle (\tilde{\rho_1},\tilde{\rho_2},...,\tilde{\rho_n}), (M_1,M_2,...,M_n)\rangle\\
&\text{subject to:} \ \  \Phi (M_1,M_2,...,M_n)\triangleq \sum_{i=1}^n{M_i}=\mathbb{I}_D\\
&(M_1,M_2,...,M_n)\in \operatorname{Pos}(\mathcal{X})^n\\
\end{aligned}
\end{equation}

\begin{equation}
\begin{aligned}
&\underline{\text { Dual problem }} \\
&\text { minimize: } \ \  \langle \mathbb{I}_D, Y\rangle\\
&\text{subject to:} \ \ \Phi^{\dagger} (Y)\succeq (\tilde{\rho_1},\tilde{\rho_2},...,\tilde{\rho_n})\\
& Y\in  \operatorname{Herm}(\mathcal{X})
\end{aligned}
\end{equation}


\providecommand{\noopsort}[1]{}\providecommand{\singleletter}[1]{#1}%

\end{document}